\documentclass[5p]{elsarticle}
\usepackage{amsmath,amssymb,amsfonts,amsthm,latexsym,url,xspace,algorithmic}
\usepackage{color}
\usepackage{graphicx}

\newcommand{\NP}{{\bf NP}}

\newcommand{\remove}[1]{}

\newcommand{\OPT}{\ensuremath{\mathsf{OPT}}\xspace}
\newcommand{\eps}{\epsilon}

\newtheorem{theorem}{\bf Theorem}[section]
\newtheorem{corollary}[theorem]{\bf Corollary}
\newtheorem{lemma}[theorem]{\bf Lemma}
\newtheorem{fact}[theorem]{\bf Fact}
\newtheorem{definition}[theorem]{\bf Definition}

\begin{document}

\title{Center-based Clustering under Perturbation Stability\tnoteref{t1}}
\tnotetext[t1]{This work was supported in part by the National Science Foundation under grant CCF-0830540, as well as by CyLab at Carnegie Mellon under grants DAAD19-02-1-0389 and W911NF-09-1-0273 from the Army Research Office.}
\author{Pranjal Awasthi}
\ead{pawasthi@cs.cmu.edu}
\address{Carnegie Mellon University, Pittsburgh, PA 15213-3891}
\author{Avrim Blum}
\ead{avrim@cs.cmu.edu}
\address{Carnegie Mellon University, Pittsburgh, PA 15213-3891}
\author{Or Sheffet}
\ead{osheffet@cs.cmu.edu}
\address{Carnegie Mellon University, Pittsburgh, PA 15213-3891}


\begin{abstract}
Clustering under most popular objective functions is NP-hard, even to
approximate well, and so unlikely to be efficiently solvable in the
worst case. Recently, Bilu and Linial~\cite{Bilu09} suggested an
approach aimed at bypassing this computational barrier by using
properties of instances one might hope to hold in practice.  In particular, they
argue that instances in practice should be stable to small
perturbations in the metric space and give an efficient algorithm for
clustering instances of the Max-Cut problem that are stable to
perturbations of size $O(n^{1/2})$.  In addition, they conjecture that
instances stable to as little as $O(1)$ perturbations should be
solvable in polynomial time.  In this paper we prove that this
conjecture is true for any center-based clustering objective (such as
$k$-median, $k$-means, and $k$-center).  Specifically, we show we can
efficiently find the optimal clustering assuming only stability to
factor-3 perturbations of the underlying metric in spaces without
Steiner points, and stability to factor $2+\sqrt{3}$ perturbations for general
metrics.  In particular, we show for such instances that the popular
Single-Linkage algorithm combined with dynamic programming will find
the optimal clustering.  We also present NP-hardness results under a
weaker but related condition.
\end{abstract}

\maketitle

\section{Introduction}

Problems of clustering data arise in a wide range of different areas
-- clustering proteins by function, clustering documents by topic, and
clustering images by who or what is in them, just to name a few.  In
this paper we focus on the popular class of center based clustering
objectives, such as $k$-median, $k$-center and $k$-means.  Under these
objectives we not only partition the data into $k$ subsets, but we
also assign $k$ special points, called the {\em centers}, one in each
cluster.  The quality of a solution is then measured as a function of
the distances between the data points and their centers.  For example,
in the $k$-median objective, the goal is to minimize the sum of
distances of all points from their nearest center, and in the
$k$-means objective, we minimize the sum of the same distances
squared.  As these are \NP-hard problems~\cite{Guha98,Jain02,
Dasgupta08}, there has been substantial work on approximation
algorithms~\cite{Arora98,Arya01,Bartal01,Charikar99,Kumar04,delaVega03}
with both upper and lower bounds on approximability of these and other
objective functions. Note that we are especially
interested in the case that $k$ is part of the input and {\em not} a
constant.

Recently, Bilu and Linial \cite{Bilu09}, focusing on the Max-Cut
problem~\cite{Garey90}, proposed considering instances where the optimal clustering
is optimal not only under the given metric, {\em but also under any
bounded multiplicative perturbation of the given metric}.  This is 
motivated by the fact that in practice, distances between data points
are typically just the result of some heuristic measure (e.g.,
edit-distance between strings or Euclidean distance in some feature
space) rather than true ``semantic distance'' between objects.  Thus,
unless the optimal solution on the given distances is correct by pure
luck, it likely is correct on small perturbations of the
given distances as well.  
Bilu and Linial \cite{Bilu09} analyze Max-Cut instances of this
type
and show that for instances that are stable to perturbations of
multiplicative factor roughly $O(n^{1/2})$, one can retrieve the
optimal Max-Cut in polynomial time.  However, they conjecture that
stability up to only {\em constant} magnitude perturbations should be
enough to solve the problem in polynomial time. In this paper we show
that this conjecture is indeed true for $k$-median and $k$-means
objectives and in fact for any well-behaved center-based objective
function (see Definition~\ref{def:separable-objective}).

\subsection{Main Result}

First, let us formally define the notion due to \cite{Bilu09} of
stability under multiplicative perturbations, stated in this context.

\begin{definition}
Given a metric $(S, d)$, and $\alpha > 1$, we say a function
$d':S\times S \rightarrow \mathbb{R}_{\geq 0}$ is an
\emph{$\alpha$-perturbation} of $d$, if for any $x,y\in S$ it holds
that \[d(x,y) \leq d'(x,y) \leq \alpha d(x,y)\] 
\end{definition}
Note that $d'$ may be any non-negative function, and need not be a
metric.

\begin{definition}
Suppose we have a clustering instance composed of $n$ points residing
in a metric $(S,d)$ and an objective function $\Phi$  we wish to
optimize. We call the  clustering instance \emph{$\alpha$-perturbation
resilient} for $\Phi$ if for any $d'$ which is an $\alpha$-perturbation
of $d$, the (only) optimal clustering of $(S,d')$ under $\Phi$ is 
identical, as a partition of points into subsets, to the optimal
clustering of $(S,d)$ under $\Phi$.
\end{definition}

We will in particular be concerned with {\em separable}, {\em
center-based} clustering objectives $\Phi$ (which include $k$-median,
$k$-means, and $k$-center among others).

\begin{definition}
\label{def:separable-objective}
A clustering objective is {\em center-based} if the optimal solution
can be defined by $k$ points $c^*_1, \ldots, c^*_k$ in the metric
space called {\em centers} such that every data point is assigned to
its nearest center.  Such a clustering objective is {\em separable} if
it furthermore satisfies the following two conditions: 
\begin{itemize}
\item The objective function value of a given clustering is either a
(weighted) sum or the maximum of the individual cluster scores. 
\item Given a proposed single cluster, its score can be computed in
polynomial time. 
\end{itemize}
\end{definition}

Our main result is that we can efficiently find the optimal
clustering for perturbation-resilient instances of separable
center-based clustering objectives.  In particular, we get an
efficient algorithm for 3-perturbation-resilient instances when the
metric $S$ is defined only over data points,
and for $(2 + \sqrt{3})$-perturbation-resiliant instances
for general metrics.

\begin{theorem}
\label{thm-3-perturbation-resilience}
For $\alpha \geq 3$ (in the case of finite metrics defined only over
the data) or $\alpha \geq 2+\sqrt{3}$ (for general metrics), there is
a polynomial-time
algorithm that finds the optimal clustering of $\alpha$-perturbation
resilient instances for any given separable center-based clustering
objective.
\end{theorem}

The algorithm, described in Section~\ref{subsec:algorithm}, turns out
to be quite simple.  As a first step, it runs the classic
single-linkage algorithm, but unlike the standard approach of halting
when $k$ clusters remain, it runs the algorithm until {\em all} points
have been merged into a single cluster and keeps track of the entire
tree-on-clusters produced.\footnote{The example depicted in
Figure~\ref{fig-single-linkage-till-k-not-enough} proves that indeed,
halting the Single-Linkage algorithm once $k$ clusters are formed may
fail on certain $\alpha$-perturbation resilient instances.}  Then, the
algorithm's second step is to apply dynamic programming to this
hierarchical clustering to identify the best $k$-clustering that is
present within the tree.  Using a result of Balcan et
al.~\cite{Balcan08} we show that the resulting clustering obtained is indeed the optimal one.
Albeit being very different, our approach resembles, in spirit, the work
of Bartal~\cite{Bartal98},
Abraham et al~\cite{ABCDGKNS05} and R\"{a}cke~\cite{Racke08} in the sense that we
reduce the problem of retrieving an optimal solution from a
general instance to a tree-like instance (where it is poly-time solvable).

Our algorithms use only a weaker property, which we call
center-proximity (see Section~\ref{subsec:properties}), that is
implied by perturbation-resilience.  We then complement these results
with a lower bound showing that for the problem of $k$-median on
general metrics, for any $\eps > 0$, there exist \NP-hard instances
that satisfy $(3-\epsilon)$-center proximity.\footnote{We note that
while our belief was that allowing Steiner points in the lower bound
was primarily a technicality, Balcan et al.~(M.F. Balcan, personal
communication) have recently shown this is not the case, giving a
clever algorithm that finds the optimal clustering for $k$-median
instances in finite metrics when $\alpha = 1+\sqrt{2}$.}

\subsection{Related work}
There have been a number of investigations of different notions of
stability for the problem of clustering.  For example, Ostrovsky et
al.~\cite{Ostrovsky06} consider a $k$-means instance to be stable if
the optimal $k$-clustering is substantially cheaper than the optimal
$(k-1)$-clustering under this objective.  They present an efficient
algorithm for finding near-optimal $k$-means clusterings when this gap
is large, and these results were subsequently strengthened to apply to
smaller gaps in \cite{ABS10}.  Balcan et al.~\cite{Balcan09}
consider instead a clustering instance to be stable if good approximations to
the given objective are guaranteed to be close, as clusterings, to a
desired ground-truth partitioning.  This is motivated by the fact that
when the true goal is to match some unknown
correct answer (e.g., to correctly cluster proteins by their
function), this is an implicit assumption already being made when
viewing approximation ratio as a good performance measure.  Balcan et
al.~\cite{Balcan09} show that in fact this condition can be used to
bypass approximation hardness results for a number of clustering
objectives including $k$-median and $k$-means.  Here they show that
if all $(1+\alpha)$-approximations to the objective are $\delta$-close
to the desired clustering in terms of how points are partitioned, then
one can efficiently get $O(\delta/\alpha)$-close to the desired
clustering.  Ben-David et al.~\cite{Ben-David06,Ben-David07} consider
a notion of stability of a clustering {\em algorithm}, which is called
stable if it outputs similar clusters for different sets of $n$ input
points drawn from the same distribution. For $k$-means, the work of
Meila~\cite{Meila06} discusses the opposite direction -- classifying
instances where an approximated solution for $k$-means is close to the
target clustering.


\section{Proof of Main Theorem}
\label{sctn-alpha-perturbation}

\subsection{Properties of Perturbation Resilient Instances}
\label{subsec:properties}

We begin by deriving other properties which every $\alpha$-perturbation
resilient clustering instance must satisfy.

\begin{definition}
\label{def:center-proximity}
Let $p\in S$ be an arbitrary point, let $c^*_i$ be the center $p$ is assigned to in the optimal clustering, and let $c^*_j\neq c^*_i$ be any other center in the optimal clustering. We say a clustering instance satisfies the \emph{$\alpha$-center proximity} property if for any $p$ it holds that \[d(p,c^*_j) > \alpha d(p,c^*_i)\]
\end{definition}

\begin{fact}
\label{fact-dist-to-other-centers}
If a clustering instance satisfies the $\alpha$-perturbation resilience property, then it also satisfies the $\alpha$-center proximity property.
\end{fact}
\begin{proof}
Let $C^*_i$ and $C^*_j$ be any two clusters in the optimal clustering and pick any $p \in C^*_i$. Assume we blow up all the pairwise distances within cluster $C^*_i$ by a factor of
$\alpha$. As this is a legitimate perturbation of the metric, it still
holds that the optimal clustering under this perturbation
is the same as the original optimum. Hence, $p$ is still assigned to
the same cluster. Furthermore, since the distances within $C^*_i$ were
all changed by the same constant factor, $c_i^*$ will still remain an
optimal center of cluster $i$. The same
holds for cluster $C^*_j$. It follows that even in this
perturbed metric, $p$ prefers $c^*_i$ to $c^*_j$. Hence $\alpha d(p,
c^*_i) = d'(p,c^*_i) < d'(p,c^*_j) = d(p,c^*_j)$.
\end{proof}

\begin{corollary}
\label{cor-distance-pts-different-clusters}
For every point $p$ and its center $c^*_i$, and for every point $p'$ from a different cluster, it follows that $d(p,p') > (\alpha-1) d(p,c^*_i)$.
\end{corollary}
\begin{proof}
Denote by $c^*_j$ the center of the cluster that $p'$ belongs to. Now, consider two cases. Case (a): $d(p',c^*_j) \geq d(p,c^*_i)$. In this case, by traingle inequality we get that  $d(p,p') \geq d(p',c^*_i) - d(p,c^*_i)$. Since the data instance is stable to $\alpha$-perturbations, Fact~\ref{fact-dist-to-other-centers} gives us that $d(p',c_i^*) > \alpha d(p',c_j^*)$. Hence we get that $d(p,p') > \alpha d(p',c^*_j) - d(p,c^*_i)$ $\geq (\alpha-1)d(p,c^*_i)$. Case (b): $d(p',c^*_j) < d(p,c^*_i)$.  Again by traingle inequality we get that $d(p,p') \geq d(p,c^*_j)-d(p',c^*_j) > \alpha d(p,c^*_i) - d(p',c^*_j) > (\alpha-1)d(p,c^*_i)$.
\end{proof}

A key ingredient in the proof of Theorem~\ref{thm-3-perturbation-resilience} is the
{\em tree-clustering} formulation of Balcan et.~al \cite{Balcan08}.  In
particular, we prove that if an instance satisfies $\alpha$-center proximity
for $\alpha \geq 3$ (in the case of finite metrics without Steiner
points) or for $\alpha \geq 2+\sqrt{3}$ (for general metrics) then it
also satisfies the ``min-stability property'' (defined below). The
min-stability property, as shown in \cite{Balcan08}, is sufficient
(and necessary) for the Single-Linkage algorithm to produce
a tree such that the optimal clustering is some pruning of this
tree. In order to define the ``min-stability'' property, we first introduce the following notation. For any two subsets $A, B \subset S$, we denote the minimum distance between $A$ and $B$ as $d_{\min}(A,B) = \min \{ d(a,b) \ |\ a\in A, b\in B\}$.

\begin{definition}
\label{def:min-stability}
A clustering instance satisfies the
\emph{min-stability property} if for any two clusters $C$ and $C'$ in
the optimal clustering, and any subset $A \subsetneq C$,
it holds that $d_{\min} (A, C\setminus A) \leq
d_{\min}(A,C')$.
\end{definition}
In words, the min-stability property means that for any set $A$ that
is a strict subset of some cluster $C$ in the optimal clustering, the closest point to
$A$ is a point from $C\setminus A$, and not from some other cluster.
The next two lemmas lie at the heart of our algorithm. 

\begin{lemma}
\label{lem-inner-cluster-attraction}
A clustering instance in which centers must be data points that
satisfies $\alpha$-center proximity for $\alpha \geq 3$ (for a
center-based clustering objective), also 
satisfies the min-stability property.
\end{lemma}
\begin{proof}
Let $C^*_i, C^*_j$ be any two clusters in the target clustering. Let
$A$ and $A'$ be any two subsets s.t.~$A\subsetneq C^*_i$ and
$A'\subseteq C^*_j$.  Let $p \in A$ and $p' \in A'$ be the two points
which obtain the minimum distance $d_{\min}(A,A')$. Let $q\in
C^*_i\setminus A$ be the nearest point to $p$. Also, denote by $c^*_i$
and $c^*_j$ the centers of clusters $C^*_i$ and $C^*_j$ respectively.

For the sake of contradiction, assume that $d_{\min}(A, C^*_i\setminus
A) \geq d_{\min}(A, A')$.    
Suppose $c^*_i\notin A$. This means that $d(p,p') = d_{\min}(A,A')
\leq d_{\min}(A, C^*_i\setminus A) \leq d(p,c^*_i)$. As $\alpha \geq
3$, this contradicts
Corollary~\ref{cor-distance-pts-different-clusters}. 

Thus we may assume $c^*_i\in A$.  It follows that $d(q,c^*_i)
\geq d(p,p') > (3-1)d(p,c^*_i) = 2d(p,c^*_i)$, so $d(p,c^*_i) <
d(q,c^*_i)/2$. We therefore have that $d(p',c^*_i) \leq d(p, p') +
d(p,c^*_i) \leq 3d(q,c^*_i)/2$. This implies that $d(p',c^*_j)$ $<$
$d(p',c^*_i)/\alpha$ $<$ $d(q,c^*_i)/2$, and thus $d(q,c^*_j) \leq
d(q,c^*_i) + d(c^*_i,p) + d(p,p') + d(p',c^*_j) < 3d(q,c^*_i) \leq
\alpha d(q,c^*_i)$. This contradicts
Fact~\ref{fact-dist-to-other-centers}. 
\end{proof}

\begin{lemma}
\label{lem-general-metric}
A clustering instance in which centers need not be data points that
satisfies $\alpha$-center proximity for $\alpha \geq 2+\sqrt{3}$ (for a
center-based clustering objective), also
satisfies the 
min-stability property. 
\end{lemma}

\begin{proof}
As in the proof of Lemma \ref{lem-inner-cluster-attraction}, let
$C^*_i, C^*_j$ be any two clusters in the target clustering and let 
$A$ and $A'$ be any two subsets s.t.~$A\subsetneq C^*_i$ and
$A'\subseteq C^*_j$.  Let $p \in A$ and $p' \in A'$ be the two points
which obtain the minimum distance $d_{\min}(A,A')$ and let $q\in
C^*_i\setminus A$ be the nearest point to $p$.   Also, as in the proof
of Lemma \ref{lem-inner-cluster-attraction}, let $c^*_i$
and $c^*_j$ denote the centers of clusters $C^*_i$ and $C^*_j$
respectively (though these need not be datapoints). 

By definition of center-proximity, we have the following inequalities:
\begin{eqnarray*}
d(p,p')+d(p',c_j^*) &\! > \! & \alpha d(p,c_i^*) \mbox{~~~[c.p.~applied to
$p$]}\\
d(p,p')+d(p,c_i^*) &\! > \! & \alpha d(p',c_j^*) \mbox{~~[c.p.~applied to
$p'$]}\\
\lefteqn{\hspace*{-1.05in} d(p,p') + d(p',c_j^*) + d(p,q) \  > \ \alpha (d(q,p) -
d(p,c_i^*))}\\
\lefteqn{\mbox{ \hspace*{-1.0in} [center proximity applied to $q$ and triangle ineq.]}}
\end{eqnarray*}
Multiplying the first inequality by $1 - \frac{1}{\alpha+1} -
\frac{1}{\alpha-1}$, the second by $\frac{1}{\alpha+1}$, the third
by $\frac{1}{\alpha-1}$, and summing them together we 
get  
$$\textstyle d(p,p') > \frac{\alpha^2 - 4\alpha+1}{\alpha-1} d(p,c_i^*) + d(q,p),$$ 
which for $\alpha = 2 + \sqrt{3}$ implies $d(p,p') > d(q,p)$ as desired.
\end{proof}

\subsection{The Algorithm}
\label{subsec:algorithm}

As mentioned, Balcan et al~\cite{Balcan08} proved (Theorem~$2$) that
if an instance satisfies min-stability, then the tree on clusters
produced by the single-linkage algorithm contains the optimal
clustering as some $k$-pruning of it.  I.e., the tree produced by
starting with $n$ clusters of size $1$ (viewed as leaves), and at each
step merging the two clusters $C$,$C'$ minimizing $d_{\min}(C,C')$
(viewing the merged cluster as their parent)
until only one cluster remains.  Given the structural results proven
above, our algorithm (see
Figure~\ref{fig:the-algorithm}) simply uses this clustering tree and
finds the best $k$-pruning using dynamic programming.

\begin{figure*}[ht]
\fbox{
\begin{minipage} {0.95\textwidth} 
\begin{enumerate}
\item Run Single-Linkage until only one cluster remains,
producing the entire tree on clusters.
\item Find the best $k$-pruning of the tree by dynamic programming using the equality 
$$
\textrm{best-}k\textrm{-pruning}(T) = \min_{0<k'<k} \{ \textrm{best-}k'\textrm{-pruning}(T\textrm{'s left child}) \ + \textrm{best-}(k-k')\textrm{-pruning}(T\textrm{'s right child}) \}
$$
\end{enumerate} 
\end{minipage}
}
\caption{\label{fig:the-algorithm} \small{Algorithm to find the
optimal $k$-clustering of instances satisfying $\alpha$-center
proximity.  The algorithm is described for the case (as in $k$-median or
$k$-means) that $\Phi$ defines the overall score to be a sum over
individual cluster scores. If it is a maximum (as in $k$-center) then
replace ``$+$'' with ``max'' above.}}
\end{figure*}

\begin{proof}[Proof of Theorem~\ref{thm-3-perturbation-resilience}]
By Lemmas~\ref{lem-inner-cluster-attraction} and
\ref{lem-general-metric}, the data satisfies the
min-stability property, which as shown in \cite{Balcan08} is
sufficient to guarantee that some pruning of the single-linkage
hierarchy is the target clustering. We then find the 
optimal clustering using dynamic programming by examining
$k$-partitions laminar with the single-linkage clustering tree. The
optimal $k$-clustering of a tree-node is either the entire subtree as
one cluster (if $k=1$), or the minimum over all choices of
$k_1$-clusters over its left subtree and $k_2$-clusters over its right
subtree (if $k>1$). Here $k_1, k_2$ are positive integers, such that
$k_1 +k_2 = k$. Therefore, we just traverse the tree bottom-up,
recursively solving the clustering problem for each
tree-node.  By assumption that the clustering objective is separable,
so each step including the base-case can be performed in polynomial
time.  For the case of $k$-median in a finite metric, for example, one
can maintain 
a $n\times O(n)$ table for all possible centers and all
possible clusters in the tree, yielding a running time of $O(n^2+nk^2)$.
For the case of $k$-means in Euclidean space, one can compute the cost
of a single cluster by computing the center as just the average of all
its points.  In general, the overall running time is $O(n(k^2 +
T(n)))$, where $T(n)$ denotes the time it takes to compute the cost of
a single cluster.
\end{proof}


\subsection{Some Natural Barriers}
\label{subsec:limitations-of-our-approach}

We complete this section with a discussion of barriers of our
approach. First, our algorithm indeed fails on some
finite metrics that are $(3-\epsilon)$-perturbation resilient.  For
example, consider the instance shown in
Figure~\ref{fig:single-linkage-fail}. In this instance, the clustering
tree produced by single-linkage is not laminar with the optimal
$k$-median clustering. It is easy to check that this instance is
resilient to $\alpha$-perturbations for any $\alpha < 3$. 

\begin{figure}[h]
\centering
\includegraphics{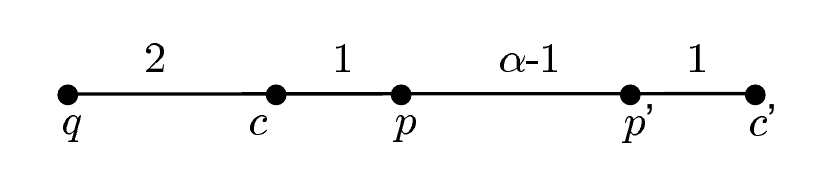}
\caption{\label{fig:single-linkage-fail} \small{A finite metric
$k$-median instance with $2< \alpha < 3$ where our algorithm
fails. The optimal $2$-median clustering is $\{c,p,q\},\{c',p'\}$. In
contrast, when we run our algorithm over on this instance, single
linkage first connects $\{c,p\}$ with $\{c',p'\}$, and only then
merges these $4$ points with $q$.}}
\end{figure}

Second, observe that our analysis, though emanating from perturbation
resilience, only uses center proximity.  We
next show that for general metrics, one cannot hope
to solve (in poly-time) $k$-median instances satisfying
$\alpha$-center proximity for $\alpha<3$.  This is close to our upper
bound of $2 + \sqrt{3}$ for general metrics.

\begin{theorem}
\label{thm:NP-hard-center-proximity}
For any $\alpha<3$, the problem of solving $k$-median instances over
general metrics that satisfy $\alpha$-center proximity is \NP-hard.
\end{theorem}
\begin{proof}
The proof of Theorem~\ref{thm:NP-hard-center-proximity} follows from
the classical reduction of Max-$k$-Coverage to $k$-median. In this
reduction, we create a bipartite graph where the right-hand side
vertices represent the elements in the ground set; the left-hand side
vertices represent the given subsets; and the distance between the
set-vertex and each element-vertex is $1$, if the set contains that
element. Using shortest-path distances, it follows that the distance
from any element-vertex to a set-vertex to which it does not belong to
is at least $3$. Using the fact that the \NP-hardness results for
Max-$k$-Coverage holds for disjoint sets (i.e. the optimal solution of
Yes-instances is composed of $k$ disjoint sets, see~\cite{Feige98}),
the $\alpha$-center proximity property follows. 
\end{proof}


Lastly, we comment that using Single-Linkage in the usual way (namely, stopping
when there are $k$ clusters remaining) is {\em not} sufficient to
produce a good clustering. We demonstrate this using the example shown in Figure~\ref{fig-single-linkage-till-k-not-enough}. Observe, in this instance, since $C$ contains significantly less points than $A$,$B$, or $D$, this instance is stable -- even if we perturb distances by a factor of $3$, the cost of any alternative clustering is higher than the cost of the optimal solution. 
However, because $d(A,C) > d(B,D)$, it follows that the usual version of Single-Linkage will unite $B$ and $D$, and only then $A$ and $C$. Hence, if we stop the Single-Linkage algorithm at $k=3$ clusters, we will not get the desired clustering.

\begin{figure}[h]
\centering
\includegraphics[scale=0.5]{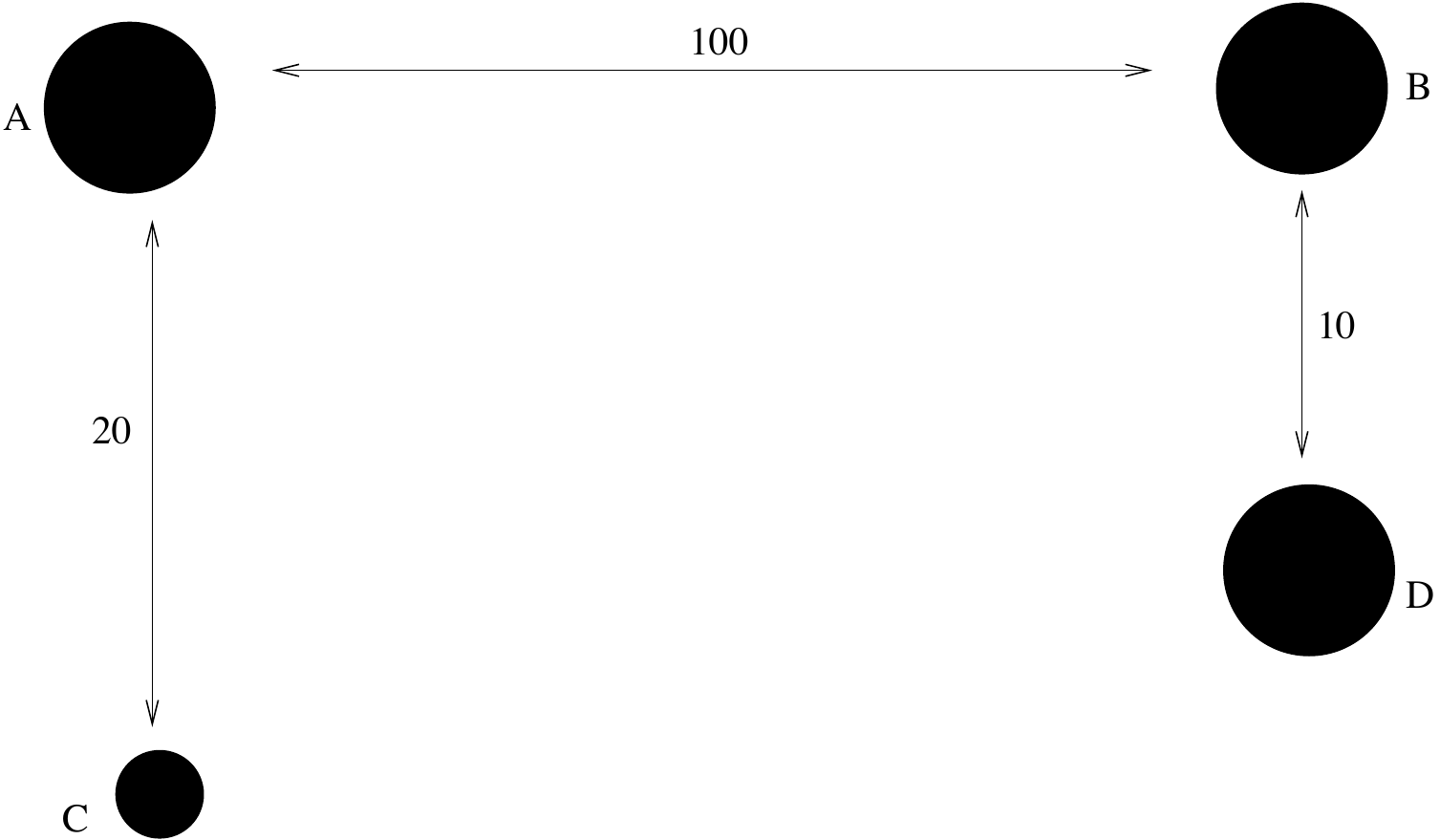}
\caption{\label{fig-single-linkage-till-k-not-enough} \small{An
example showing failure of the usual version of Single-Linkage. The
instance is composed of $4$ components, each with inner-distance
$\eps$ and outer-distance as described in the figure. However,
components $A,B$ and $D$ each contain $100$ points, whereas component
$C$ has only $10$ points. The optimal $3$-median clustering consists
of $3$ clusters: $\{A,C\},\{B\},\{D\}$ and has cost $\OPT = 200 +
300\eps$.}}  
\end{figure}

\section{Open Problems}
\label{sctn-open-problems}

There are several natural open questions left by this work.  First,
can one reduce the perturbation factor $\alpha$ needed for efficient
clustering?  As mentioned earlier, recently Balcan et
al.~(M.F.~Balcan, personal communication) have given a very
interesting algorithm that reduces the $\alpha=3$ factor needed by our
algorithm for finite metrics to $1+\sqrt{2}$.  Can one go farther,
perhaps by using further implications of perturbation-resilience
beyond center-proximity?   Alternatively, if one cannot find the {\em
optimal} clustering for small values of $\alpha$, can one
still find a near-optimal clustering, of approximation ratio better
than what is possible on worst-case instances?

In a different direction, one can also consider relaxations of the
perturbation-resilience condition.  For example, Balcan et
al.~(personal communication) also consider instances that are ``mostly
resilient'' to $\alpha$-perturbations: under
any $\alpha$-perturbation of the underlying metric, no more than a
$\delta$-fraction of the points get mislabeled under the optimal
solution.  For sufficiently large constant $\alpha$ and sufficiently
small constant $\delta$, they present
algorithms that get good approximations to the objective under this
condition.  A different kind of relaxation would be to consider a 
notion of \emph{resilience to
perturbations on average}: a clustering instance whose optimal
clustering is likely not to change, assuming the perturbation is {\em
random} from some suitable distribution.  Can this weaker notion be used to
still achieve positive guarantees?  


\bibliographystyle{plain}
\bibliography{paper}

\end{document}